\documentclass[12pt]{amsart}
\usepackage{amsmath,amssymb}

\newtheorem{thm}{Theorem}[section]

\newtheorem{lem}[thm]{Lemma}
\newtheorem{cor}[thm]{Corollary}

\theoremstyle{definition}

\theoremstyle{remark}

\numberwithin{equation}{section}

\begin{document}


\title{A linear time algorithm for a variant of the max cut problem in series parallel graphs}


\author{Brahim Chaourar}
\address{Department of Mathematics and Statistics, Al Imam Mohammad Ibn Saud Islamic University (IMSIU), P.O. Box
90950, Riyadh 11623,  Saudi Arabia \\ Correspondence address: P.O. Box 287574, Riyadh 11323, Saudi Arabia}
\email{bchaourar@hotmail.com}





\begin{abstract}
Given a graph $G=(V, E)$, a connected sides cut $(U, V\backslash U)$ or $\delta (U)$ is the set of edges of E linking all vertices of U to all vertices of $V\backslash U$ such that the induced subgraphs $G[U]$ and $G[V\backslash U]$ are connected. Given a positive weight function $w$ defined on $E$, the maximum connected sides cut problem (MAX CS CUT) is to find a connected sides cut $\Omega$ such that $w(\Omega)$ is maximum. MAX CS CUT is NP-hard. In this paper, we give a linear time algorithm to solve MAX CS CUT for series parallel graphs. We deduce a linear time algorithm for the minimum cut problem in the same class of graphs without computing the maximum flow.
\end{abstract}


\maketitle





\section{Introduction}

Sets and their characterisitic vectors will not be distinguished. We refer to Bondy and Murty \cite{Bondy and Murty 2008} about graph theory terminolgy and facts.
\newline Given an undirected graph $G = (V, E)$ and positive weights $w_{ij} = w_{ji}$ on the edges $(i, j)\in E$, the maximum cut problem (MAX CUT) is that of finding the set of vertices $S$ that maximizes the weight of the edges in the cut $(S, V\backslash S)$ or $\delta (S)$ or $\delta (V\backslash S)$; that is, the weight of the edges with one endpoint in $S$ and the other in $V\backslash S$. The (decision variant of the) MAX CUT is one of the Karp’s original NP-complete problems \cite{Karp 1972}, and has long been known to be NP-complete even if the problem is unweighted; that is, if $w_{ij} = 1$ for all $(i, j)\in E$ \cite{Garey et al. 1976}. This motivates the research to solve the MAX CUT problem in special classes of graphs. The MAX CUT problem is solvable in polynomial time for the following special classes of graphs: planar graphs \cite{Barahona 1990, Hadlock 1975, Orlova and Dorfman 1972}, line graphs \cite{Guruswami 1999}, graphs with bounded treewidth, or cographs \cite{Bodlaender and Jansen 2000}. But the problem remains NP-complete for chordal graphs, undirected path graphs, split graphs, tripartite graphs, graphs that are the complement of a bipartite graph \cite{Bodlaender and Jansen 2000} and planar graphs if the weights are of arbitrary sign \cite{Terebenkov 1991}. Besides its theoretical importance, the MAX CUT problem has applications in circuit layout design and statistical physics \cite{Barahona et al. 1988}. For a comprehensive survey of the MAX CUT problem, the reader is referred to Poljak and Tuza \cite{Poljak and Tuza 1995} and Ben-Ameur et al. \cite{Ben-Ameur et al. 2014}. The best known algorithm for MAX CUT in planar graphs has running time complexity $O(n^{3/2} log n)$, where $n$ is the number of vertices of the graph \cite{Shih et al. 1990}. The main result of this paper is to exhibit a linear time algorithm for a special variant of MAX CUT in series parallel graphs.
\newline Let us give some definitions. Given an undirected graph $G = (V, E)$ and a subset of vertices $U$, a connected sides cut $\delta (U)$ is a cut where both induced subgraphs $G[U]$ and $G[V\backslash U]$ are connected. Special connected sides cuts are trivial cuts, i.e. cuts with one single vertex in one side. The corresponding weighted variant of MAX CUT for connected sides cuts is called MAX CONNECTED SIDES CUT problem (MAX CS CUT). It is clear that MAX CUT and MAX CS CUT are the same problem for complete graphs. Since MAX CUT is NP-hard for complete graphs (see \cite{Karp 1972}) then MAX CS CUT is NP-hard in the general case. Another motivation is that MAX CS CUT gives a lower bound for MAX CUT.
\newline A parallel closure of a graph is an induced subgraph on two vertices. A series extension of the graph $G = (V, E)$ based on the edge $e\in E$ is adding a vertex $v$ of degree 2 in the middle of $e$ in order to have two edges instead of $e$. A parallel extension of $G$ based on the edge $e$ is adding an edge $f$ having the same incident vertices as $e$. Series parallel graphs are graphs obtained by applying recursively series and/or parallel extensions starting from one edge. A series degree of a vertex $v$ in a graph $G$ is the degree of $v$ after replacing every parallel closure of $G$ by one single edge. A series labeling of the vertices of a series parallel graph is a labeling of the vertices from 0 to $n-1 = |V|-1$ starting from the first two vertices $v_0$ and $v_1$ and so on to the last added vertex. Any series parallel graph contains at least one vertex of series degree 2. So, given a vertex $v$ of series degree 2 with the two parallel closures $P_0$ and $P_1$ incident to $v$, and the two adjacent vertices $u_0$ and $u_1$ to $v$, we can contract all edges of $P_0$ (or $P_1$) and replace $v$ by $u_0$ (or $u_1$), and we obtain a new series parallel graph with a new vertex of series degree 2. Each involved graph in any step of this process is labeled $G_j, 0 \leq j\leq n-1$, with $G_{n-1} = G$ and $G_1$ is the induced subgraph on the two vertices $v_0$ and $v_1$.
\newline Let $G_1$ and $G_2$ be two graphs with $e_j$ an edge of $G_j, j = 1, 2$. The 2-sum of $G_1$ and $G_2$, denoted $G_1\oplus_e G_2$, based on the edges $e_1$ and $e_2$ is the graph obtained by identifying $e_1$ and $e_2$ on an edge $e$, and keeping $G_j/e_j, j = 1, 2$, as it is.
\newline We say that MAX CS CUT is linear for a class of graphs if there is a linear time algorithm to solve it in such class.
\newline The remaining of the paper is organized as follows: in section 2, we give a linear time algorithm for MAX CS CUT in series parallel graphs, in section 3, we prove that 2-sums preserve the linearity of MAX CS CUT. We deduce a linear time algorithm for MIN CUT in series parallel graphs in section 4, and we conclude in section 5.

\section{MAX CS CUT is linear for series parallel graphs}

MAXCSCUTSP Algorithm:
\newline Input: A series parallel graph $G = (V, E)$ with a series labeling of $V$, a positive weight function $w$ defined on $E$.
\newline Output: A $w$-maximum connected sides cut $\Omega$ in $G$.
\newline 0) Begin
\newline 1) $j := n-1$;
\newline 2) While $j > 1$ do
\newline 3) Begin
\newline 4) Let $P_0$ and $P_1$ be the two parallel closures incident to $v_j$ in $G_j$:
\newline 5) If $w(P_0) > w(P_1)$ then contract $P_1$;
\newline 6) Else: contract $P_0$;
\newline 7) $j := j-1$;
\newline 6) End of While
\newline 7) $j := 2$;
\newline 8) $\Omega := E(G_1)$;
\newline 9) While $j\leq n-1$ do
\newline 10) Begin
\newline 11) Let $P_0$ and $P_1$ the two parallel closures incident to $v_j$ in $G_j$:
\newline 12) If $w(P_0)+w(P_1)>w(\Omega)$ then $\Omega := P_0\cup P_1$;
\newline 13) $j := j+1$;
\newline 14) End of While
\newline 15) End of MAXCSCUTSP algorithm.
\newline This algorithm has two phases: Phase I (steps 1-6) and Phase II (steps 7-14). In each step, we do roughly $n$ operations, so the complexity of MAXCSCUTSP is $O(n)$, where $n=|V|$.
\begin{thm}
MAXCSCUTSP algorithm solves MAX CS CUT in series parallel graphs.
\end{thm}
\begin{proof}
The summary of the algorithm is as follows:
MAXCSCUT chooses a vertex v with series degree 2 (step 4) and contract the less weighted parallel closure incident to v (steps 5 and 6). And so on the resulted graph until it reaches $G_1$ , the starting single parallel closure (Phase I). In $G_1$, the $w$-maximum connected sides cut is $E(G_1)$ (step 8). After that, it goes in the reverse path (Phase II): the $w$-maximum connected sides cut is either the trivial cut based on the current vertex $v_j$ with series degree 2 or the current computed connected sides cut (step 12).
Let $v_j$ be the chosen vertex with series degree 2 in $G_j$, $P_0$ and $P_1$ the two parallel closures incident to $v_j$. Without loss of generality, we can suppose that $w(P_0)<w(P_1)$ and $G_{j-1}=G_j/P_0$. Let $\Omega_j$ be the w-maximum connected sides cut in $G_j, 1\leq j\leq n-1$. It suffices to prove that $w(\Omega_j) = Max \{w(\Omega_{j-1}), w(P_0\cup P_1)\}$.
\newline Let $\Omega$ be a connected sides cut in $G_j$ distinct from $P_0\cup P_1$. Since $w(P_0)<w(P_1)$, we have only two cases:
\newline \textbf{Case 1:} $P_1\subseteq \Omega$ then $\Omega$ is a connected sides cut in $G_{j-1} = G_j/P_0$ containing $P_1$. And vice versa, any connected sides cut in $G_{j-1} = G_j/P_0$ containing $P_1$ is a connected sides cut in $G_j$ containing $P_1$.
\newline \textbf{Case 2:} $P_1\nsubseteq \Omega$ then $\Omega$ is a connected sides cut in $G_{j-1}=G_j/P_0$ not containing $P_1$. And vice versa, any connected sides cut in $G_{j-1}=G_j/P_0$ not containing $P_1$ is a connected sides cut in $G_j$ not containing $P_1$.
\newline So the connected sides cuts candidates for the $w$-maximum connected sides cut in $G_j$ and $G_{j-1}$ are the same, except $P_0\cup P_1$.
\end{proof}
Note that MAXCSCUT algorithm solves MAX CS CUT in series parallel graphs even for arbitrary sign weight functions.

\section{2-sums preserve linearity of MAX CS CUT}

Let $\mathcal C(G)$ be the class of connected sides cuts of $G$. We need the following lemma.
\begin{lem}
$\mathcal C(G_1\oplus_e G_2)=\{ \Omega_j\in \mathcal C(G_j) : e_j\notin \Omega_j, j=1, 2\}\cup \{ \Omega_1\oplus_e \Omega_2 : \Omega_j\in \mathcal C(G_j)$ and $e_j\in \Omega_j, j = 1, 2\}$.
\end{lem}
It follows that a $w$-maximum connected sides cut in $G_1\oplus_e G_2$ is one of the three following connected sides cuts:
\newline (cases 1-2) one of the two $w$-maximum connected sides cuts in $G_j$ which does not contain $e_j, j = 1, 2$,
\newline (case 3) or the 2-sum of the $w$-maximum connected sides cuts containing $e_j, j = 1, 2$.
\newline To find a $w$-maximum connected sides cut in $G_j$ which does not contain $e_j, j = 1, 2$ (case 2), we have to contract $e_j$. We need then to perform at most $c(n_j-1)$ operations, where $c$ is the linearity coefficient and $n_j, j=1, 2$ is the number of vertices of $G_j$ (by induction).
\newline To find $\Omega_1\oplus_e \Omega_2$ (case 3), we have to put $w(e_j), j = 1, 2$, as big as possible, e.g. sum of the positive weights of all edges, and find $\Omega_j, j = 1, 2$. In this case, we need to perform at most  $c(n_1+n_2)$ operations (by induction).
\newline So we have to compute MAX CS CUT twice in each graph and compare three cuts. The total number of operations is bounded then by $2c(n_1+n_2-1)=2c(n-1)$, where $n$ is the number of vertices of $G_1\oplus_e G_2$. So linearity of the problem is preserved.

\section{MIN CUT is linear for series parallel graphs}

MINCUTSP Algorithm:
\newline Input: A series parallel graph $G = (V, E)$ with a series labeling of $V$, a positive weight function $w$ defined on $E$.
\newline Output: A $w$-minimum connected sides cut $\Omega$ in $G$.
\newline We keep the same steps as MAXCSCUTSP algorithm except the following changes in two steps:
\newline 5) If $w(P_0) < w(P_1)$ then contract $P_1$;
\newline 12) If $w(P_0)+w(P_1)<w(\Omega)$ then $\Omega := P_0\cup P_1$;
\newline Since this algorithm is similar to MAXCSCUTSP, then its complexity is $O(n)$, where $n=|V|$.
\newline And it is not difficult to see, similarly to MAXCSCUTSP, that MINCUTSP gives the minimum weighted connected sides cut in a series parallel graph without computing the maximum flow.
\newline We can conclude with the following result.
\begin{thm}
Given a connected graph $G=(V, E)$ and a positive weight function $w$ defined on $E$. Then any $w$-minimum cut is a connected sides cut of $G$.
\end{thm}
\begin{proof}
Let $\delta (U)$ be a cut with $G[U]$ disconnected. It suffices to prove that $\delta (U)$ is not a $w$-minimum cut. Let $G[U_1]$ be one connected component of $G[U]$. Since $G$ is connected, then $w(V\backslash U, U_1)>0$ (i.e. there are edges between $V\backslash U$ and $U_1$). It follows that $w(\delta (U\backslash U_1))=w(\delta (U))-w(V\backslash U, U_1)<w(\delta (U))$.
\end{proof}
Another consequence of Lemma 3.1 and Theorem 4.1 is the following corollary.
\begin{cor}
2-sums preserves the linearity of MIN CUT.
\end{cor}

\section{Conclusion}
We have introduced a new variant of MAX CUT: MAX CS CUT, which is also NP-hard. We have provided two linear time algorithms for MAX CS CUT and MIN CUT, respectively, in series parallel graphs.
We have proved that 2-sums preserve the linearity of MAX CS CUT and MIN CUT.
Further directions are to study MAX CS CUT in larger classes of graphs than series parallel graphs.

\noindent{ \bf Acknowledgements}
The author  is grateful to the deanship of Scientific Research at Al Imam Mohammad Ibn Saud Islamic University (IMSIU) for supporting financially this research under the grant No 331203.

\end{document}